\documentclass[prl,aps,floatfix,superscriptaddress,notitlepage,twocolumn]{revtex4-1}
\usepackage[utf8]{inputenc}

\usepackage{fullpage} 
\usepackage{tikz} 
\usepackage{amsmath,amsfonts,amsthm, mathtools, mathrsfs}
\usepackage{float}
\usepackage{dsfont}
\usepackage{csquotes}
\usepackage{amssymb}
\usepackage{verbatim}
\usepackage{enumerate}

\usepackage{graphicx}
\usepackage[leftcaption]{sidecap}

\usepackage{hyperref}
\definecolor{darkred}  {rgb}{0.5,0,0}
\definecolor{darkblue} {rgb}{0,0,0.5}
\definecolor{darkgreen}{rgb}{0,0.5,0}
\hypersetup{
  colorlinks = true,
  urlcolor  = blue,         
  linkcolor = darkblue,     
  citecolor = darkgreen,    
  filecolor = darkred       
}

\theoremstyle{definition}
\newtheorem*{corollary}{Corollary}
\newtheorem{definition}{Definition}

\newtheorem{lemma}{Lemma}
\newtheorem{proposition}{Proposition}
\newtheorem{theorem}{Theorem}

\newcommand{\mc}{\mathcal}

\newcommand{\1}{\mathds{1}}

\newcommand{\one}{\mathds{1}}

\definecolor{cool_green}{rgb}{0.0, 0.5, 0.0}
\definecolor{cool_blue}{rgb}{0.0, 0.0, 0.5}

\newcommand{\Tr}[1]{\operatorname{Tr}\!\left[#1\right]}
\def\>{\rangle}
\def\<{\langle}

\def\mE{\mathcal{E}}

\def\sH{\mathcal{H}}

\newcommand{\set}[1]{\mathcal{#1}}

\newcommand{\pguess}{P_{\operatorname{guess}}}

\renewcommand{\qedsymbol}{\nobreak \ifvmode \relax \else
	\ifdim \lastskip<1.5em \hskip-\lastskip \hskip1.5em plus0em
	minus0.5em \fi \nobreak \vrule height0.75em width0.5em
	depth0.25em\fi}

\renewcommand{\ge}{\geqslant}
\renewcommand{\le}{\leqslant}

\begin{document}

\title{A Complete Resource Theory of Quantum Incompatibility\\as Quantum Programmability}
\author{Francesco Buscemi}
\email{buscemi@i.nagoya-u.ac.jp}
\affiliation{Graduate  School  of  Informatics,  Nagoya
	University, Chikusa-ku, 464-8601 Nagoya, Japan}
\author{Eric Chitambar}
\email{echitamb@illinois.edu}
\affiliation{Department of Electrical and Computer Engineering, Coordinated Science Laboratory, University of Illinois at Urbana-Champaign, Urbana, IL 61801}
\author{Wenbin Zhou}
\email{zhou.wenbin@i.mbox.nagoya-u.ac.jp}
\affiliation{Graduate  School  of  Informatics,  Nagoya
	University, Chikusa-ku, 464-8601 Nagoya, Japan}
\date{\today}

\begin{abstract}
Measurement incompatibility describes two or more quantum measurements whose expected joint outcome on a given system cannot be defined.  This purely non-classical phenomenon provides a necessary ingredient in many quantum information tasks such violating a Bell Inequality or nonlocally steering part of an entangled state.  In this paper, we characterize incompatibility in terms of programmable measurement devices and the general notion of quantum programmability.  This refers to the temporal freedom a user has in issuing programs to a quantum device.  For devices with a classical control and classical output, measurement incompatibility emerges as the essential quantum resource embodied in their functioning.  Based on the processing of programmable measurement devices, we construct a quantum resource theory of incompatibility.  A complete set of convertibility conditions for programmable devices is derived based on quantum state discrimination with post-measurement information.  
 
\end{abstract}

\maketitle

The theory and practice of quantum measurement is a topic that sits at the foundation of quantum mechanics.  Unlike its classical counterpart, quantum measurement offers a variety of ways to probe a system and extract classical information.  A highly non-classical feature that emerges in quantum mechanics is measurement incompatibility.  The most general quantum measurements are described by positive-operator valued measures (POVMs), and incompatibility of POVMs is typically defined in terms of joint measurability \cite{Lahti-2003a, Heinosaari-2008a, Heinosaari2016}.  Roughly speaking, a family of POVMs is called jointly measurable if the outcomes of the constituent POVMs can be simulated through the measurement of a single ``mother'' POVM.

There has been much interest in measurement incompatibility and its relationship to various primitive tasks in quantum information theory~\cite{Heinosaari2016}.  For the demonstration of quantum nonlocality, it is not difficult to see that a Bell Inequality can be violated only if incompatible measurements are employed by each of the parties involved in the experiment \cite{Fine-1982a}.  While for certain families of measurements the converse is true \cite{Wolf-2009a}, only  recently has it been found not to hold in general \cite{Quintino-2016a, Bene-2018a}.  However, this asymmetry between measurement incompatibility and nonlocality vanishes when considering the more general task of quantum steering.  That is, a family of POVMs is incompatible if and only if it can be used to steer some quantum state in a non-classical way \cite{Uola-2014a, Quintino-2014a}.  In recent works, it was also shown that a family of POVMs is incompatible if and only if it offers an advantage in some state discrimination game \cite{Carmeli2018,Carmeli-2019a,Uola-2019a,Mori-2019}.  The main result of this letter offers a generalization of these results.

Given the ability of incompatible measurements to generate non-classical effects and enhance quantum state discrimination tasks, it becomes natural to view measurement incompatibility as a resource in quantum information processing.  This interpretation can be formalized using the framework of a quantum resource theory (QRT) \cite{Chitambar-2019a}.  In general, a QRT isolates some particular feature of a quantum system referred to as a resource, such as entanglement or coherence, and studies how this resource transforms under a restricted set of ``free'' operations; crucially, the free operations cannot generate the resource on their own.  While entanglement and coherence represent static resources that are commonly studied in the literature, it is also possible to formulate resource theories for dynamic resources such as certain families of quantum measurements \cite{Oszmaniec-2017a, Takagi-2019a, Oszmaniec-2019a, Designolle2019}.  

In particular, resource theories of measurement incompatibility have been previously proposed in which the resources are incompatible families of POVMs \cite{Heinosaari-2015a, Guerini-2017a, Skrzypczyk-2019a}.  However, a drawback to these approaches is that the free operations identified are not large enough to fully capture the notion of measurement incompatibility in an operational way.  Ref. \cite{Heinosaari-2015a} only considers measurement convertibility under quantum pre-processing while Refs. \cite{Guerini-2017a, Skrzypczyk-2019a} only consider conditional classical post-processing as the free operations.  Both of these on their own are too weak in that do not allow for the free convertibility of one compatible POVM family to another.  Moreover, there is no \textit{a priori} reason why an experimenter should be restricted to performing either just quantum pre- or classical post-processing when their combination is equally unable to generate incompatibility.

In this letter, we construct a resource theory of measurement incompatibility that combines both quantum pre-processing and conditional classical post-processing in the context of programmable measurement devices (PMDs).  PMDs are objects that emerge through the following consideration.  In any experiment where different measurements are being employed, there are two relevant systems: the quantum system $Q$ that is subjected to the particular measurement and the ``program'' system whose state $x\in\mc{X}$ represents the choice of measurement.  The measurement apparatus in such an experiment thus exemplifies a PMD since the type of measurement it performs depends on the program it receives.
 
To formulate a resource theory in this setting, we shift the primary focus away from quantum measurement and place it on \textit{programmability}, which we consider broadly to be any sort of classical control over a device that can be implemented at the programmer's discretion.  In other words, we envision programmability to mean that some device can be obtained at time $t_0$ and then controlled in whatever way the device allows at some later time $t$.  This reflects the natural interplay between computing hardware and software: one first purchases or builds a computing device and then later programs it to perform whatever computational task is desired.  However, adopting such a perspective then requires constraining the type of interaction between the program and quantum system described in the previous paragraph.  Namely, the program system should not be allowed to affect the preparation of the quantum system since the former is decided at time $t$ while the latter is set at time $t_0<t$.  In satisfying this restriction, we are thus lead to a resource theory of programmability for which the free operations arise from very natural physical considerations.

Let us now put the discussion in more formal terms.  
\begin{definition}[Programmable Measurement Devices]\label{def:PMD}
	A (classically) \textit{programmable measurement device} (PMD) is a collection of POVMs on the same Hilbert space $\sH^Q$, $\{M^Q(a|x):a\in\mc{A},x\in\mc{X}\}$ such that $M^Q(a|x)\geq 0$ and $\sum_a M^Q(a|x)=\1^Q$ for all $x$. The set $\mc{X}$ is interpreted as the program set (an element $x$ being the program), while the set $\mc{A}$ is interpreted as the outcome set.
\end{definition}
\noindent 
While PMDs are mathematically equivalent to $\text{cq}\to \text{c}$ channels, the two inputs of a PMD are always assumed to be \textit{separate} systems.  Crucially, we assume that it takes a finite amount of time for the program to be able to influence the measurement performed on the quantum system.

This assumption immediately implies a necessary condition for a PMD to be able to implement an incompatible family of POVMs: the PMD must be able to effectively preserve the quantum system at least for the time it takes the program to influence the measurement process.  This simple observation leads us to define the free objects in our QRT as those PMDs corresponding to compatible families of POVMs.
\begin{definition}[Simple PMDs, alias Compatible POVMs]\label{def:simple}
A PMD $M^Q(a|x)$ is called \textit{simple} if its constituting POVMs can be written as
\begin{equation}
\label{Eq:simple-PMD}
M^Q(a|x)=\sum_{i\in \mc{I}}p(a|i,x) \tilde{M}^Q(i),
\end{equation}
where the $\tilde{M}^Q(i)$ are elements of a single POVM (sometime referred to as the ``mother'' POVM), and $p(a|i,x)$ is a conditional probability distribution.
\end{definition}

Compatible POVMs are also often defined in terms of coarse-graining over a single POVM, and Eq.~\eqref{Eq:simple-PMD} is equivalent to this characterization (see, e.g., Ref.~\cite{Guerini-2017a}).  From their definition, simple PMDs can be perfectly simulated without the need to store the quantum system, which can be immediately measured using the mother POVM, with the program influencing only the classical post-processing of its outcome.

Let us then turn to the free operations in this resource theory, which will be a restricted set of maps converting $\text{cq}\to \text{c}$ channels to $\text{cq}\to \text{c}$.  Such maps convert channel $\mc{M}$ into $\mc{F}_{\text{post}}\circ\mc{M}\circ\mc{E}_{\text{pre}}$, where $\mc{E}_{\text{pre}}$ and $\mc{F}_{\text{post}}$ are pre- and post-processing maps, possibly connected by a memory side channel \cite{Chiribella-2008a, Gour-2019a}.  Every non-simple PMD functions as a quantum memory as it must preserve the quantum system until the program arrives.  Quantum memory, then, is essential for a device to be programmable, and so it should not be something freely available in a resource theory of programmability.  The memory connecting $\mc{E}_{\text{pre}}$ and $\mc{F}_{\text{post}}$ should therefore be classical, and the pre-processing map $\mc{E}_{\text{pre}}$ should be causally independent of the program since at that time the program has not arrived yet.  What remains are the free operations of this QRT, and they are described by Eq. \eqref{eq:strategies} in the following definition (see Fig.~\ref{Fig:PMD-temporal-processing} for a schematic representation).

\begin{definition}\label{def:PMD-proc}
	Given two PMDs $M^Q(a|x)$ and $N^{Q'}(b|y)$ on $\sH^Q$ and $\sH^{Q'}$ respectively, we write $M^Q(a|x)\succeq N^{Q'}(b|y)$ whenever
	\begin{align}\label{eq:strategies}
		N^{Q'}(b|y)=\sum_r&\mu(r)\sum_{i,x,a}q(b|a,x,i,y,r)\times\notag\\
		&p(x|i,y,r)
		(\mE^{Q'\to Q}_{i|r})^\dagger [M^Q(a|x)]\;,
	\end{align}
	where (i) $\mu(r)$ is a probability distribution modeling a shared source of classical randomness, (ii) $\{\mE^{Q'\to Q}_{i|r} \}$ is a family of quantum instruments labeled by $r$, with classical outcome $i$, and $\mc{E}^\dagger$ denotes the adjoint (i.e., trace-dual) map of $\mc{E}$, and (iii) $p(x|i,y,r)$ and $q(b|a,x,i,y,r)$ are classical noisy channels (conditional probability distributions).  The relation $M^Q(a|x)\succeq N^{Q'}(b|y)$ expresses convertibility of PMDs by free operations in this QRT.
\end{definition}
	
	\begin{figure}[t]
	\centering
	\includegraphics[width=\columnwidth]{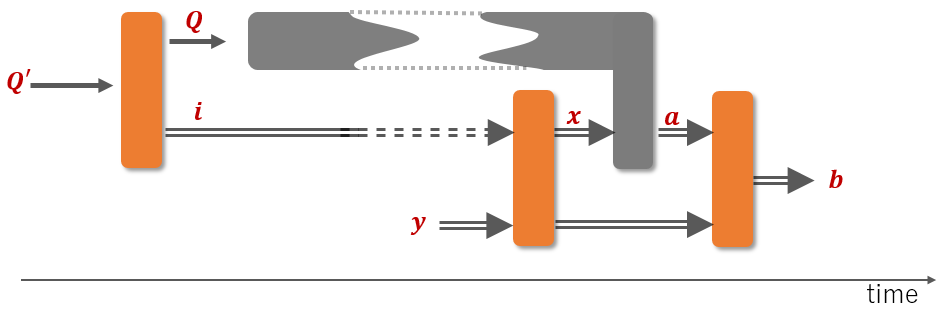}
			\caption{PMDs processing, according to Definition~\ref{def:PMD-proc}. Time flows from left to right. The program $y$ (i.e., the post information) arrives after the pre-processing has been performed.  Since the only quantum memory resides in the PMD, the quantum input must be committed to the PMD until the program arrives.  On the other hand, the classical output $i$ of the pre-processing instrument can be stored in a classical memory and interact with the program before it reaches the PMD.  Notice that, even though it is not explicitly depicted in the picture, classical randomness can be shared between all processing boxes (orange on-line), so that the set of possible processings is convex.}
		\label{Fig:PMD-temporal-processing}
	\end{figure}

Before proceeding further, we stress that the free operations considered here need not constitute the \textit{only} meaningful operational framework to study the properties of programmability and compatibility.  However, as shown in the Supplemental Material, they do satisfy the important property that any two simple PMDs can always be freely interconverted.

We refer to Fig. \ref{Fig:PMD-temporal-processing} as the temporal model of PMD processing, and there is an alternative spatial model that characterizes PMDs in terms of bipartite channels shared between two spatially separated parties (Alice and Bob).  As shown in Fig. \ref{fig:povm-processing}, the programmability of a PMD is then translated into a no-signaling constraint from Bob to Alice.  Hence the correct operational setting for PMD processing in the spatial model is one-way LOCC, and the following proposition makes this connection precise.
\begin{proposition}
$M^Q(a|x)\succeq N^{Q'}(b|y)$ if and only if $M^Q(a|x)$ can be converted to $N^{Q'}(b|y)$ by a one-way LOCC from Alice to Bob.
\end{proposition}

\noindent We stress that while the bipartite processing of PMDs by one-way LOCC is intuitively simple, without the temporal model in mind, the physical motivation for studying the QRT of $\text{cq}\to \text{c}$ channels under one-way LOCC is less clear.  Why is one-way LOCC the free set of operations in such a QRT, and why must it only be from Alice to Bob?  The answers come from the allowed operations in the temporal model, which do have clear physical motivation in terms of programmability.  It just so happens that these free operations correspond to Alice-to-Bob one-way LOCC in the spatial model.

\begin{figure}[t]
	\centering
	\includegraphics[width=\columnwidth]{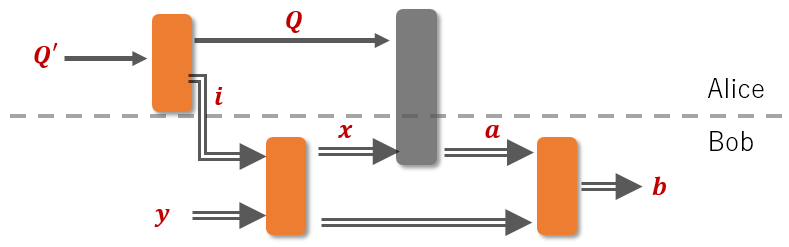}
	\caption{The spatial model of PMD processing.  The quantum and program inputs are separated between Alice and Bob, and the free operations depicted in Fig. \ref{Fig:PMD-temporal-processing} translate into one-way LOCC maps from Alice to Bob.}
	\label{fig:povm-processing}
\end{figure}

\textit{PMDs and Post-Information Guessing Games.}  The main result of this paper is a characterization of free PDM convertibility in terms of quantum state guessing games with side information \cite{Ballester-2008a,Carmeli-2019a,Skrzypczyk-2019a,Uola-2019a}.  These games involve a referee who distributes to the player a quantum state and some classical side information, information which we will henceforth refer to as ``post information'' since it more appropriately fits our temporal model.  More formally, let $\{\rho^R_{w,z}:w\in\mc{W},z\in\mc{Z} \}$ be a two-index quantum ensemble such that $p(w,z):=\Tr{\rho^R_{w,z}}$ is a normalized joint probability distribution.   A post-information guessing game consists of the following components: (i) the referee picks one pair $(w,z)\in\mc{W}\times\mc{Z}$ at random according to the distribution $p(w,z)$, (ii) the normalized quantum state $p(w,z)^{-1}\rho^R_{w,z}$ is sent to the player followed, after some finite time, by the index $w$, and  (iii) the player attempts to maximize the probability of correctly guessing the value $z$ using the given PMD $M^Q(a|x)$ and any free processing described in Definition \ref{def:PMD-proc}.	In this game, the label $w$ is interpreted as the ``post information'' since it is imported into the program register of the PMD \textit{after} the quantum state, and it cannot be used in any pre-processing of the PMD.

When playing guessing games with post information, certain processing strategies will lead to greater success probabilities in guessing $z$.  In particular, if the referee's questions $\rho^R_{w,z}$ are encoded on a quantum system that is different from the quantum input of the PMD $M^Q(a|x)$, then the player must do some sort of quantum pre-processing of $R$ into $Q$, represented without loss of generality by a quantum instrument $\{\mc{E}_i^{R\to Q}\}$.  The optimum success probability over all strategies is thus given by
\begin{widetext}
\begin{equation}\label{eq:optimal-guess}
\pguess(M^Q(a|x);\rho^R_{w,z}):=\max_{\mu,q,p,\mE}\sum_{w,z}\sum_r\sum_{i,x,a}\mu(r)q(z|a,w,i,r)p(x|w,i,r)\ \Tr{\mE^{R\to Q}_{i|r}(\rho^R_{w,z})\ M^Q(a|x)}\;,
\end{equation}
\end{widetext}
where the probability distribution $\mu(r)$ is included to describe mixed strategies, i.e. those in which a different strategy, labeled by $r$, is chosen at random. (The optimum guessing probability will then be achieved on pure strategies, but it is convenient to explicitly include this in Eq.~(\ref{eq:optimal-guess}).)

We are now ready to state the main result, whose proof, which closely follows those in~\cite{Buscemi2016} and~\cite{Buscemi2017}, is given in the Supplemental Material.
\begin{theorem}\label{th:main}
	Given two PMDs $M^Q(a|x)$ and $N^{Q'}(b|y)$, the following are equivalent:
	\begin{enumerate}[(a)]
		\item $M^Q(a|x)\succeq N^{Q'}(b|y)$;
		\item for all guessing games with post-information $\{\rho^R_{w,z}:w\in\mc{W},z\in\mc{Z} \}$, \[\pguess(M^Q(a|x);\rho^R_{w,z})\ge \pguess(N^{Q'}(b|y);\rho^R_{w,z})\;.\]
	\end{enumerate}
	In~(b), it is possible to consider only guessing games with $\sH^R=\sH^{Q'}$, $\set{W}=\set{Y}$, and $\set{Z}=\set{B}$.
\end{theorem}

Simply by noticing that it is impossible to turn a simple PMD into an incompatible one by means of free operations, we obtain as a corollary that quantum incompatibility can always be witnessed by means of a suitable guessing game with post-information.

\begin{corollary}
	A PMD $M^Q(a|x)$ is incompatible, if and only if there exists an ensemble $\{\rho^Q_{x,a}: x\in\mc{X},a\in\mc{A} \}$ such that
	\[
	\sum_{a,x}\Tr{M^Q(a|x)\ \rho^Q_{x,a}}>\pguess^{\text{simple}}(\rho^Q_{x,a})\;,
	\]
	where $\pguess^{\text{simple}}(\rho^Q_{x,a})$ is defined as the optimum guessing probability achievable with simple PMDs.
\end{corollary}

As a special case, Theorem \ref{th:main} provides necessary and sufficient conditions of a single POVM under quantum pre-processing and conditional post-processing.  That is, we have $M^Q(a)\succ N^{Q'}(b)$ in the sense of Eq. \eqref{eq:strategies} if and only if the POVM $M^Q(a)$ is always more useful than  $N^{Q'}(b)$ for the task of minimum-error state discrimination, i.e. for every ensemble $\{\rho_z\}_z$ we have $\pguess(M^{Q}(b);\rho^{R}_{z})\geq\pguess(N^{Q'}(b);\rho^{R}_{z})$.

\textit{Robustness of non-Simple PMDs.}  In any QRT with minimal structure, it is possible to define a (generalized) robustness measure of resource \cite{Brandao-2015a}.  Roughly speaking, the robustness captures how tolerant an object is to mixing before it loses all its resource.  A PMD robustness measure $\mathfrak{R}( \{M(a|x)\}_{a,x})$ can also be defined in this QRT directly analogous to incompatibility robustness measures previously studied \cite{Uola-2019a, Skrzypczyk-2019a}.  Specifically, we have
\begin{align*}
&\mathfrak{R}( \{M(a|x)\}_{a,x})\\
&=\min\left\{r\geq 0:\frac{M(a|x) + rN(a|x)}{1+r} \in \mathcal{F}\right\}\;,
\end{align*}
where $\mathcal{F}$ is the convex, compact set of simple PMDs matching input and output spaces of $M(a|x)$.  In the above corollary, we showed that every incompatible PMD has an advantage over incompatible ones in some guessing game with post-information.  This advantage can also be quantified as the maximum ratio between the optimum guessing probability of the given incompatible PMD versus the optimal guess probability of any simple PMD,
\[
\max_{\{\rho^Q_{x,a} \}} \frac{\pguess(M^Q(a|x);\rho^Q_{x,a})}{\pguess^{\text{simple}}(\rho^Q_{x,a})}\;.
\]
It is possible to show an equivalence between the advantage and the robustness. Namely, 
\begin{align*}
&1+\mathfrak{R}( \{M(a|x)\}_{a,x})\\
&=\max_{\{\rho^Q_{x,a} \}} \frac{\pguess(M^Q(a|x);\rho^Q_{x,a})}{\pguess^{\text{simple}}(\rho^Q_{x,a})}\;,
\end{align*}
where the maximization is over all possible guessing games with post-information, mathematically represented by a double-index ensemble $\{\rho^Q_{x,a}: x\in\mc{X},a\in\mc{A} \}$. The proof follows in the same way as the proof of Theorem 2 in~\cite{PhysRevX.9.031053}, and one can check for details in the supplementary document.  This establishes an operational interpretation of $\mathfrak{R}( \{M(a|x)\}_{a,x})$ in terms of guessing games with post-information.

\textit{Conclusion.}  In this letter we have shown that a resource theory of quantum incompatibility can be naturally formulated as a resource theory of programmability, and that new insights can be gained by doing so. In particular, this resource theory is complete in the sense that all free devices are naturally equivalent to each other.  This was accomplished by identifying programmability as a key resource that requires quantum memory for its realization.  From this perspective, both quantum pre-processing and classical conditional post-processing can be integrated into the picture, while remaining, however, within the operational scenario provided by post-information guessing games~\cite{Carmeli2018}.

The approach that we followed here in order to formulate a resource theory of quantum incompatibility is very much inspired by the concept of statistical comparison, introduced in mathematical statistics chiefly by Blackwell~\cite{Blackwell-1953} and extended to the quantum case by one of the present authors~\cite{Buscemi2012}. Indeed, the aim of statistical comparison, as originally envisaged by Blackwell, is that of expressing the possibility of transforming an initial statistical model into another one, in terms of the utility that the two statistical models provide in operationally motivated scenarios (that is, statistical decision problems in Blackwell's original paper). \textit{Mutatis mutandis}, this is exactly the scope of any resource theory, where the aim is to identify a set of operationally motivated monotones that dictate when an allowed transformation between resources exists or not. Among the numerous examples of such an approach, which at present ranges from quantum nonlocality~\cite{Buscemi2012} to quantum thermodynamics~\cite{Gour2018}, the present work bears some similarities with the resource theory of quantum memories, viz. non entanglement-breaking channels, recently put forth in Ref.~\cite{Rosset2018}. Even though no program register is considered in~\cite{Rosset2018}, there, as it happens here, the quantum memory is probed by means of ``timed'' decision problems, in which two tokens of the problem (there, two quantum tokens; here, one token is classical) are given to the player at subsequent times, which is then asked to formulate an educated guess so to maximize the expected payoff. Further relations between the two frameworks are left for future research.

\begin{acknowledgments}
We thank Teiko Heinosaari, Gilad Gour, and Otfried G\"{u}hne for helpful discussions on measurement incompatibility. W.Z. thanks Ryuji Takagi and Bartosz Regula for pointing out some mistakes while writing the proof of robustness. F.B. acknowledges partial  support  from  the  Japan Society for the Promotion of Science (JSPS) KAKENHI,  Grant  No.19H04066.  E.C. acknowledges the support of NSF Award No. 1914440. W.Z. is supported by the Program (RWDC) for Leading Graduate Schools of Nagoya University.
\end{acknowledgments}

\bibliographystyle{alphaurl}
\bibliography{programmability}

\newcommand{\etalchar}[1]{$^{#1}$}
\begin{thebibliography}{OGWA17}

\bibitem[BG15]{Brandao-2015a}
Fernando G. S.~L. Brand{\~a}o and Gilad Gour.
\newblock Reversible framework for quantum resource theories.
\newblock {\em Phys. Rev. Lett.}, 115:070503, Aug 2015.
\newblock \href {https://doi.org/10.1103/PhysRevLett.115.070503}
  {\path{doi:10.1103/PhysRevLett.115.070503}}.

\bibitem[Bla53]{Blackwell-1953}
David Blackwell.
\newblock Equivalent comparisons of experiments.
\newblock {\em The Annals of Mathematical Statistics}, 24(2):265--272, 1953.
\newblock URL: \url{http://www.jstor.org/stable/2236332}.

\bibitem[Bus12]{Buscemi2012}
Francesco Buscemi.
\newblock {Comparison of Quantum Statistical Models: Equivalent Conditions for
  Sufficiency}.
\newblock {\em Communications in Mathematical Physics}, 310(3):625--647, jan
  2012.
\newblock \href {https://doi.org/10.1007/s00220-012-1421-3}
  {\path{doi:10.1007/s00220-012-1421-3}}.

\bibitem[Bus16]{Buscemi2016}
Francesco Buscemi.
\newblock {Degradable channels, less noisy channels, and quantum statistical
  morphisms: an equivalence relation}.
\newblock {\em Problems of Information Transmission}, 52(3):201--213, jul 2016.
\newblock \href {https://doi.org/10.1134/s0032946016030017}
  {\path{doi:10.1134/s0032946016030017}}.

\bibitem[Bus17]{Buscemi2017}
Francesco Buscemi.
\newblock Comparison of noisy channels and reverse data-processing theorems.
\newblock In {\em 2017 {IEEE} Information Theory Workshop ({ITW})}. {IEEE}, nov
  2017.
\newblock \href {https://doi.org/10.1109/itw.2017.8278038}
  {\path{doi:10.1109/itw.2017.8278038}}.

\bibitem[BV18]{Bene-2018a}
Erika Bene and Tam{\'{a}}s V{\'{e}}rtesi.
\newblock Measurement incompatibility does not give rise to bell violation in
  general.
\newblock {\em New Journal of Physics}, 20(1):013021, jan 2018.
\newblock \href {https://doi.org/10.1088/1367-2630/aa9ca3}
  {\path{doi:10.1088/1367-2630/aa9ca3}}.

\bibitem[BWW08]{Ballester-2008a}
M.~A. {Ballester}, S.~{Wehner}, and A.~{Winter}.
\newblock State discrimination with post-measurement information.
\newblock {\em IEEE Transactions on Information Theory}, 54(9):4183--4198,
  2008.
\newblock \href {https://doi.org/10.1109/TIT.2008.928276}
  {\path{doi:10.1109/TIT.2008.928276}}.

\bibitem[CDP08]{Chiribella-2008a}
G.~Chiribella, G.~M. D'Ariano, and P.~Perinotti.
\newblock Transforming quantum operations: Quantum supermaps.
\newblock {\em {EPL} (Europhysics Letters)}, 83(3):30004, jul 2008.
\newblock URL: \url{https://doi.org/10.1209%2F0295-5075%2F83%2F30004}, \href
  {https://doi.org/10.1209/0295-5075/83/30004}
  {\path{doi:10.1209/0295-5075/83/30004}}.

\bibitem[CG19]{Chitambar-2019a}
Eric Chitambar and Gilad Gour.
\newblock Quantum resource theories.
\newblock {\em Rev. Mod. Phys.}, 91:025001, Apr 2019.
\newblock \href {https://doi.org/10.1103/RevModPhys.91.025001}
  {\path{doi:10.1103/RevModPhys.91.025001}}.

\bibitem[CHT18]{Carmeli2018}
Claudio Carmeli, Teiko Heinosaari, and Alessandro Toigo.
\newblock State discrimination with postmeasurement information and
  incompatibility of quantum measurements.
\newblock {\em Physical Review A}, 98(1), jul 2018.
\newblock \href {https://doi.org/10.1103/physreva.98.012126}
  {\path{doi:10.1103/physreva.98.012126}}.

\bibitem[CHT19]{Carmeli-2019a}
Claudio Carmeli, Teiko Heinosaari, and Alessandro Toigo.
\newblock Quantum incompatibility witnesses.
\newblock {\em Phys. Rev. Lett.}, 122:130402, Apr 2019.
\newblock \href {https://doi.org/10.1103/PhysRevLett.122.130402}
  {\path{doi:10.1103/PhysRevLett.122.130402}}.

\bibitem[DFK19]{Designolle2019}
S\'ebastien Designolle, M\'at\'e Farkas, and Jedrzej Kaniewski.
\newblock Incompatibility robustness of quantum measurements: a unified
  framework.
\newblock 2019.
\newblock \href {http://arxiv.org/abs/http://arxiv.org/abs/1906.00448v2}
  {\path{arXiv:http://arxiv.org/abs/1906.00448v2}}.

\bibitem[Fin82]{Fine-1982a}
Arthur Fine.
\newblock Joint distributions, quantum correlations, and commuting observables.
\newblock {\em Journal of Mathematical Physics}, 23(7):1306--1310, 1982.
\newblock \href {https://doi.org/10.1063/1.525514}
  {\path{doi:10.1063/1.525514}}.

\bibitem[GBCA17]{Guerini-2017a}
Leonardo Guerini, Jessica Bavaresco, Marcelo~Terra Cunha, and Antonio Ac\'{i}n.
\newblock Operational framework for quantum measurement simulability.
\newblock {\em Journal of Mathematical Physics}, 58(9):092102, 2017.
\newblock \href {https://doi.org/10.1063/1.4994303}
  {\path{doi:10.1063/1.4994303}}.

\bibitem[GJB{\etalchar{+}}18]{Gour2018}
Gilad Gour, David Jennings, Francesco Buscemi, Runyao Duan, and Iman Marvian.
\newblock Quantum majorization and a complete set of entropic conditions for
  quantum thermodynamics.
\newblock {\em Nature Communications}, 9(1), dec 2018.
\newblock \href {https://doi.org/10.1038/s41467-018-06261-7}
  {\path{doi:10.1038/s41467-018-06261-7}}.

\bibitem[{Gou}19]{Gour-2019a}
G.~{Gour}.
\newblock Comparison of quantum channels by superchannels.
\newblock {\em IEEE Transactions on Information Theory}, 65(9):5880--5904, Sep.
  2019.
\newblock \href {https://doi.org/10.1109/TIT.2019.2907989}
  {\path{doi:10.1109/TIT.2019.2907989}}.

\bibitem[HKR15]{Heinosaari-2015a}
Teiko Heinosaari, Jukka Kiukas, and Daniel Reitzner.
\newblock Noise robustness of the incompatibility of quantum measurements.
\newblock {\em Phys. Rev. A}, 92:022115, Aug 2015.
\newblock \href {https://doi.org/10.1103/PhysRevA.92.022115}
  {\path{doi:10.1103/PhysRevA.92.022115}}.

\bibitem[HMZ16]{Heinosaari2016}
Teiko Heinosaari, Takayuki Miyadera, and M{\'{a}}rio Ziman.
\newblock An invitation to quantum incompatibility.
\newblock {\em Journal of Physics A: Mathematical and Theoretical},
  49(12):123001, feb 2016.
\newblock \href {https://doi.org/10.1088/1751-8113/49/12/123001}
  {\path{doi:10.1088/1751-8113/49/12/123001}}.

\bibitem[HRS08]{Heinosaari-2008a}
Teiko Heinosaari, Daniel Reitzner, and Peter Stano.
\newblock Notes on joint measurability of quantum observables.
\newblock {\em Foundations of Physics}, 38(12):1133--1147, Dec 2008.
\newblock \href {https://doi.org/10.1007/s10701-008-9256-7}
  {\path{doi:10.1007/s10701-008-9256-7}}.

\bibitem[Lah03]{Lahti-2003a}
Pekka Lahti.
\newblock Coexistence and joint measurability in quantum mechanics.
\newblock {\em International Journal of Theoretical Physics}, 42(5):893--906,
  May 2003.
\newblock \href {https://doi.org/10.1023/A:1025406103210}
  {\path{doi:10.1023/A:1025406103210}}.

\bibitem[{Mor}19]{Mori-2019}
Junki {Mori}.
\newblock {Operational characterization of incompatibility of quantum channels
  with quantum state discrimination}.
\newblock {\em arXiv e-prints}, 2019.
\newblock \href {http://arxiv.org/abs/1906.09859} {\path{arXiv:1906.09859}}.

\bibitem[OB19]{Oszmaniec-2019a}
Micha{\l{}} Oszmaniec and Tanmoy Biswas.
\newblock Operational relevance of resource theories of quantum measurements.
\newblock {\em {Quantum}}, 3:133, April 2019.
\newblock \href {https://doi.org/10.22331/q-2019-04-26-133}
  {\path{doi:10.22331/q-2019-04-26-133}}.

\bibitem[OGWA17]{Oszmaniec-2017a}
Micha\l{} Oszmaniec, Leonardo Guerini, Peter Wittek, and Antonio Ac\'{\i}n.
\newblock Simulating positive-operator-valued measures with projective
  measurements.
\newblock {\em Phys. Rev. Lett.}, 119:190501, Nov 2017.
\newblock \href {https://doi.org/10.1103/PhysRevLett.119.190501}
  {\path{doi:10.1103/PhysRevLett.119.190501}}.

\bibitem[QBHB16]{Quintino-2016a}
Marco~T\'ulio Quintino, Joseph Bowles, Flavien Hirsch, and Nicolas Brunner.
\newblock Incompatible quantum measurements admitting a local-hidden-variable
  model.
\newblock {\em Phys. Rev. A}, 93:052115, May 2016.
\newblock \href {https://doi.org/10.1103/PhysRevA.93.052115}
  {\path{doi:10.1103/PhysRevA.93.052115}}.

\bibitem[QVB14]{Quintino-2014a}
Marco~T\'ulio Quintino, Tam\'as V\'ertesi, and Nicolas Brunner.
\newblock Joint measurability, einstein-podolsky-rosen steering, and bell
  nonlocality.
\newblock {\em Phys. Rev. Lett.}, 113:160402, Oct 2014.
\newblock \href {https://doi.org/10.1103/PhysRevLett.113.160402}
  {\path{doi:10.1103/PhysRevLett.113.160402}}.

\bibitem[RBL18]{Rosset2018}
Denis Rosset, Francesco Buscemi, and Yeong-Cherng Liang.
\newblock {Resource Theory of Quantum Memories and Their Faithful Verification
  with Minimal Assumptions}.
\newblock {\em Physical Review X}, 8(2), may 2018.
\newblock \href {https://doi.org/10.1103/physrevx.8.021033}
  {\path{doi:10.1103/physrevx.8.021033}}.

\bibitem[S{\v{S}}C19]{Skrzypczyk-2019a}
Paul Skrzypczyk, Ivan {\v{S}}upi\'{c}, and Daniel Cavalcanti.
\newblock All sets of incompatible measurements give an advantage in quantum
  state discrimination.
\newblock {\em Phys. Rev. Lett.}, 122:130403, Apr 2019.
\newblock \href {https://doi.org/10.1103/PhysRevLett.122.130403}
  {\path{doi:10.1103/PhysRevLett.122.130403}}.

\bibitem[TR19a]{Takagi-2019a}
Ryuji Takagi and Bartosz Regula.
\newblock General resource theories in quantum mechanics and beyond:
  operational characterization via discrimination tasks, 2019.
\newblock \href {http://arxiv.org/abs/arXiv:1901.08127}
  {\path{arXiv:arXiv:1901.08127}}.

\bibitem[TR19b]{PhysRevX.9.031053}
Ryuji Takagi and Bartosz Regula.
\newblock General resource theories in quantum mechanics and beyond:
  Operational characterization via discrimination tasks.
\newblock {\em Phys. Rev. X}, 9:031053, Sep 2019.
\newblock URL: \url{https://link.aps.org/doi/10.1103/PhysRevX.9.031053}, \href
  {https://doi.org/10.1103/PhysRevX.9.031053}
  {\path{doi:10.1103/PhysRevX.9.031053}}.

\bibitem[UKS{\etalchar{+}}19]{Uola-2019a}
Roope Uola, Tristan Kraft, Jiangwei Shang, Xiao-Dong Yu, and Otfried G\"uhne.
\newblock Quantifying quantum resources with conic programming.
\newblock {\em Phys. Rev. Lett.}, 122:130404, Apr 2019.
\newblock \href {https://doi.org/10.1103/PhysRevLett.122.130404}
  {\path{doi:10.1103/PhysRevLett.122.130404}}.

\bibitem[UMG14]{Uola-2014a}
Roope Uola, Tobias Moroder, and Otfried G\"uhne.
\newblock Joint measurability of generalized measurements implies classicality.
\newblock {\em Phys. Rev. Lett.}, 113:160403, Oct 2014.
\newblock \href {https://doi.org/10.1103/PhysRevLett.113.160403}
  {\path{doi:10.1103/PhysRevLett.113.160403}}.

\bibitem[WPGF09]{Wolf-2009a}
Michael~M. Wolf, David Perez-Garcia, and Carlos Fernandez.
\newblock Measurements incompatible in quantum theory cannot be measured
  jointly in any other no-signaling theory.
\newblock {\em Phys. Rev. Lett.}, 103:230402, Dec 2009.
\newblock \href {https://doi.org/10.1103/PhysRevLett.103.230402}
  {\path{doi:10.1103/PhysRevLett.103.230402}}.

\end{thebibliography}

\appendix
\newpage
\onecolumngrid

\section{Convexity of Simple PMDs}

Recall the definition of simple PMDs.
\begin{definition}[Simple PMDs, \textit{alias} Compatible POVMs]\label{def:simple}
	A PMD $M^Q(a|x)$ is called \textit{simple} if its constituting POVMs can be written as
	\begin{equation}
	\label{Eq:simple-PMD}
	M^Q(a|x)=\sum_{i\in \mc{I}}p(a|i,x) \tilde{M}^Q(i),
	\end{equation}
	where the $\tilde{M}^Q(i)$ are elements of a single POVM (sometime referred to as the ``mother'' POVM), and $p(a|i,x)$ is a conditional probability distribution.
\end{definition}
\noindent As noted in the body of the letter, any convex mixing of simple PMDs can be directly incorporated into the ``mother'' POVM.  We now describe this in a bit more detail.  Suppose that $M^Q(a|x)$ admits a decomposition of the form
\[
M^Q(a|x)=\sum_r\mu(r)\sum_{i\in \mc{I}}p(a|i,x,r) \tilde{M}^Q(i|r),
\]
where $\mu(r)$ is a probability distribution and $\tilde{M}^Q(i|r)$ is now a \textit{family} of POVMs indexed by the shared random index $r$. Then, simply by noticing that $\mu(r)\tilde{M}^Q(i|r)$ is itself a normalized two-outcome indexed POVM, it is possible to conclude that Definition~\ref{def:simple} is fully general and no further random variables are needed.

\section{Free Convertibility among Simple PMDs}

Recall that $M^Q(a|x)\succeq N^{Q'}(b|y)$ whenever
\begin{align}\label{eq:strategies}
N^{Q'}(b|y)= \sum_r\mu(r)\sum_{i,x,a}q(b|a,x,i,y,r)p(x|i,y,r)(\mE^{Q'\to Q}_{i|r})^\dagger [M^Q(a|x)]\;.
\end{align}

\begin{lemma}\label{lemma:simple-all-equiv}
	All simple devices are free, that is, given any two simple devices $M^Q(a|x)$ and $N^{Q'}(b|y)$, possibly defined on different Hilbert spaces $\sH^Q$ and $\sH^{Q'}$, both relations holds:
	\[
	M^Q(a|x)\succeq N^{Q'}(b|y)\quad\text{and}\quad N^{Q'}(b|y)\succeq M^{Q}(a|x)\;.
	\]
\end{lemma}
\begin{proof}
	For any two simple PMDs.  Let us denote by $I^Q(a|x)$ the trivial PMD, i.e. the PMD with alphabets $\set{A}=\set{X}=\{0 \}$ and Hilbert space $\sH^Q=\mathbb{C}$.  Clearly the trivial PMD can be attained from any other using the free operations.  Showing that the converse is true will complete the proof of the lemma.  Using the trivial PMD as the input PMD in Eq. \eqref{eq:strategies}, we see that the instruments $\{\mE^{Q'\to Q}_{i|r} \}$ are, in fact, POVMs $E^{Q'}(i|r)$: this is so because $\dim\sH^Q=1$. Since these POVMs can be freely chosen, all devices of the form
	\[
	\begin{split}
	N^{Q'}(b|y)&=\sum_r\mu(r)\sum_{i,j}q(b|j,r)p(j|i,y,r)\ E^{Q'}(i|r)\\
	&=\sum_r\sum_{i}p'(b|i,y,r)\ E^{Q'}(i,r)\;,
	\end{split}
	\]
	can be obtained from the trivial PMD, where $E^{Q'}(i,r):=\mu(r)E^{Q'}(i|r)$ is considered now as a POVM with two outcome indices. Since the above coincides with the definition of simple PMDs, the desired conclusion is reached.	
\end{proof}

\section{Proof of Proposition 1}

\begin{proposition}
	$M^Q(a|x)\succeq N^{Q'}(b|y)$ if and only if $M^Q(a|x)$ can be converted to $N^{Q'}(b|y)$ by a one-way LOCC from Alice to Bob.
\end{proposition}

\begin{proof}
	It is obvious from Fig. 1 in the main text that $M^Q(a|x)\succeq N^{Q'}(b|y)$ implies an implementation by one-way LOCC.  Conversely, every one-way LOCC protocol from Alice to Bob consists here of (i) a one-way LOCC pre-processing, (ii) local side channels that are quantum for Alice and classical for Bob, and (iii) one-way LOCC post-processing.  Since Alice receives no output from the PMD, any local post-processing and forward communication she performs can be included in her pre-processing.  What remains is exactly as depicted in Fig. 2.
\end{proof}

\section{Proof of Theorem 1 and its Corollary}

Recall the optimal success probability of post-information guessing games:
\begin{equation}\label{eq:optimal-guess}
\pguess(M^Q(a|x);\rho^R_{w,z}):=\max_{\mu,q,p,\mE}\sum_{w,z}\sum_r\sum_{i,x,a}\mu(r)q(z|a,w,i,r)p(x|w,i,r)\ \Tr{\mE^{R\to Q}_{i|r}(\rho^R_{w,z})\ M^Q(a|x)}\;.
\end{equation}

\begin{theorem}\label{th:main}
	Given two PMDs $M^Q(a|x)$ and $N^{Q'}(b|y)$, the following are equivalent:
	\begin{enumerate}[(a)]
		\item $M^Q(a|x)\succeq N^{Q'}(b|y)$;
		\item for all guessing games with post-information $\{\rho^R_{w,z}:w\in\mc{W},z\in\mc{Z} \}$, \[\pguess(M^Q(a|x);\rho^R_{w,z})\ge \pguess(N^{Q'}(b|y);\rho^R_{w,z})\;.\]
	\end{enumerate}
	In~(b), it is possible to consider only guessing games with $\sH^R=\sH^{Q'}$, $\set{W}=\set{Y}$, and $\set{Z}=\set{B}$.
\end{theorem}

\begin{proof}
	For the sake of notation, we will denote the processing of a PMD $M^Q(a|x)$ as prescribed in Eq.~(\ref{eq:strategies}) simply by
	\[
	[\mathcal{T}(M)](b|y)\;.
	\]
	In particular, the set of all allowed mappings of PMDs with input Hilbert space $\sH^Q$, input alphabet $\mc{X}$, and output alphabet $\mc{A}$, into PMDs with input Hilbert space $\sH^{Q'}$, input alphabet $\mc{Y}$, and output alphabet $\mc{B}$, will be denoted by $\mathscr{T}$:
	\[
	\mathscr{T}:=\{\mc{T}:\text{PMD}(\sH^Q,\mc{X},\mc{A})\to\text{PMD}(\sH^{Q'},\mc{Y},\mc{B}) \}\;.
	\]
	A crucial observation is that the set $\mathscr{T}$ is convex due to the presence of shared randomness (represented by the probability distribution $\mu(r)$ in Eq. \eqref{eq:strategies}).
	
	The implication (a)$\implies$(b) is trivial: since processings of the form~(Eq. \eqref{eq:strategies}) are always allowed when playing guessing games with post-information, as prescribed in~ Eq. \eqref{eq:optimal-guess}), if PMD $M^Q(a|x)$ can simulate $N^{Q'}(b|y)$, then any strategy that can be reached from the latter can obviously be reached also from the former. Hence, we only need to prove explicitly the implication (b)$\implies$(a).
	
	We begin by noticing that condition~(a) is equivalent to the existence of a mapping $\mc{T}$ of the form~(\ref{eq:strategies}) such that
	\begin{equation}\label{eq:cond-a}
	[\mc{T}(M)]^{Q'}(b|y)=N^{Q'}(b|y),\qquad\forall b,\forall y\;.
	\end{equation} 
	
	Let us fix a basis of self-adjoint operators $\{X_j^{Q'} :j\in\mc{J}\}$. Then, relation~(\ref{eq:cond-a}) is equivalent to the following:
	\[
	\Tr{N^{Q'}(b|y)\ X_j^{Q'}}=\Tr{[\mc{T}(M)]^{Q'}(b|y)\ X_j^{Q'} }\;,\qquad\forall b,\forall y,\forall j\;.
	\]
	Denote by $\vec{s}(N)$ the vector whose entries are the $|\set{B}|\times|\set{Y}|\times|\set{J}|$ real numbers above, and by $\vec{r}(M,\mathcal{T})$ the same vector on the right-hand side.
	
	Let us consider now the set of all such vectors that can be obtained from PMD $M^Q(a|x)$ by varying the processing $\mathcal{T}$ in $\mathscr{T}$; denote such set by
	\[
	\mathscr{S}(M):=\{\vec{r}(M,\mathcal{T}):\mathcal{T}\in\mathscr{T} \}\;.
	\]
	Such a set is closed and convex, because closed and convex is the set of all transformations $\mathscr{T}$. Hence, we can say that relation~(\ref{eq:cond-a}) is equivalent to
	\[
	\vec{s}(N)\in \mathscr{S}(M)\;,
	\]
	that is, by applying the separation theorem for convex sets,
	\[
	\vec{s}(N)\cdot\vec{c}\le\max_{\vec{r}\in \mathscr{S}(M)}\vec{r}\cdot\vec{c}\;,\qquad\forall\vec{c}\in\mathbb{R}^{|\set{B}|\times|\set{Y}|\times|\set{J}|}\;.
	\]
	
	Denoting by $Y_{b,y}^{Q'}$ the self-adjoint operators obtained as $Y_{b,y}^{Q'}:=\sum_jc(b,y,j)X_j^{Q'}$, we have that relation~(\ref{eq:cond-a}) is equivalent to
	\[
	\sum_{b,y}\Tr{N^{Q'}(b|y)\ Y_{b,y}^{Q'}}\le\max_{\mathcal{T}\in\mathscr{T}}\sum_{b,y}\Tr{[\mathcal{T}(M)]^{Q'}(b|y)\ Y^{Q'}_{b,y}}\;,\qquad\forall\{Y^{Q'}_{b,y}:\text{self-adjoint} \}\;.
	\]
	
	We now shift and rescale the self-adjoint operators $Y^{Q'}_{b,y}$ to $\rho^{Q'}_{b,y}:=\frac{Y^{Q'}_{b,y}+C}{\sum_{b,y}\Tr{Y^{Q'}_{b,y}+C}}\ge 0$, so that the $\rho^{Q'}_{b,y}$ form an ensemble. This can always be done by choosing the constant operator $C$ large enough. Then, by noticing that $\sum_{b,y}\Tr{N^{Q'}(b|y)\ C}=|\set{Y}|\Tr{C}$ does not depend on the particular PMD $N^{Q'}(b|y)$, we can rewrite the above equation arriving at the following conclusion: condition~(\ref{eq:cond-a}) is equivalent to
	\[
	\sum_{b,y}\Tr{N^{Q'}(b|y)\ \rho^{Q'}_{b,y}}\le\max_{\mathcal{T}\in\mathscr{T}}\sum_{b,y}\Tr{[\mathcal{T}(M)]^{Q'}(b|y)\ \rho^{Q'}_{b,y}}\;,\qquad\forall\ \text{ensembles}\ \{\rho^{Q'}_{b,y}:b\in\set{B},y\in\set{Y} \}\;.
	\]
	Comparing the above relation with the expression~(\ref{eq:optimal-guess}) of the optimal guessing probability in guessing games with post-information, we recognize that the above equation means that, for any guessing game with post-information $\{\rho^{Q'}_{b,y}:b\in\set{B},y\in\set{Y} \}$, it holds that
	\begin{equation}\label{eq:almost-final-passage}
	\sum_{b,y}\Tr{N^{Q'}(b|y)\ \rho^{Q'}_{b,y}}\le \pguess(M^Q(a|x);\rho^{Q'}_{b,y})\;.
	\end{equation}
	But then, a \textit{sufficient} condition for relation~(\ref{eq:cond-a}) is that
	\[
	\pguess(N^{Q'}(b|y);\rho^{Q'}_{b,y})\le\pguess(M^Q(a|x);\rho^{Q'}_{b,y})\;,
	\]
	for all guessing game with post-information $\{\rho^{Q'}_{b,y}:b\in\set{B},y\in\set{Y} \}$.		
\end{proof}

\noindent\textbf{Proof of Corollary:}

First, we notice that, for any guessing game with post-information, the optimum guessing probability is the same \textit{for all} simple PMDs. This is a direct consequence of Theorem~\ref{th:main} and Lemma~\ref{lemma:simple-all-equiv}.

Then, the statement is proved by contradiction. Suppose that, for all guessing game with post-information $\{\rho^Q_{x,a}: x\in\mc{X},a\in\mc{A} \}$, the opposite relation holds, that is
\[
\sum_{a,x}\Tr{M^Q(a|x)\ \rho^Q_{x,a}}\le\pguess^{\text{simple}}(\rho^Q_{x,a})\;.
\]
But then, by means of Eq.~(\ref{eq:almost-final-passage}) in the proof above, one would conclude that it is possible to obtain $M^Q(a|x)$ by acting with a free operation on a simple PMD, in contradiction with the fact that $M^Q(a|x)$ is incompatible.

\section{Proof of Robustness}
In this section, we prove the connection of robustness and the guessing game scenario. We show that the generalized robustness of PMD is an exact quantifer for the advantage in some guessing games.

Suppose we are given a PMD, $\{M(a|x): a\in \set{A}, x\in \set{X}\}$  on the Hilbert space $\mathcal{H}^Q$, and an ensemble $\{\rho_{a,x}\}$.  According to the theorem in main text, it is possible to restrict all $\rho_{a,x} \in \mathcal{L}(\mathcal{H}^Q) \quad \forall a,x $. 

At first, we define the set of \textit{simple} PMDs as,
\begin{equation}
\set{F}_{Q,\set{A},\set{X}}:=\{ \{M(a|x)\}_{a,x}:\exists \text{ POVM } \{\widetilde{M}^Q(i)\},  p(a|x,i), \text{s.t. }M^Q(a|x)=\sum_{i\in \mc{I}}p(a|i,x) \widetilde{M}^Q(i) \quad \forall a,x\},
\end{equation}

Definition~(\ref{Eq:simple-PMD}) identifies $\set{F}_{Q,\set{A},\set{X}}$ as a collection of POVM familes that is convex and closed. Since it is possible to fix $Q, \set{A}, \set{X}$, in what follows we ignore the subscripts. We define $\mathds{M}:= \{M(a|x)\}_{a,x}$ and in what follows we use the same font style to represent PMDs. We denote by $\mathcal{Z}$ the set of general PMDs, that is to say $\mathcal{F}\subseteq \mathcal{Z}$. 

For usefullness, we define a real vector space $ \mathcal{V}$ as,  
\begin{equation}\label{vector}
\mathcal{V}: =\Bigg\{ \mathds{V}=
\scalebox{0.5}[1]{$\displaystyle
	\left(\scalebox{2}[1]{$\displaystyle
		\begin{array}{cc}
		V_{1}    \\
		\vdots   \\
		V_{d} 
		\end{array}$}
	\right)$} : V_i=V^{\dagger}_i \quad \forall i
\Bigg\}
\end{equation}
on $\mathbb{R}_{+}$, while we define its inner product as $\langle\mathds{A}|\mathds{B}\rangle:=\sum_{i}\langle A_i,B_i\rangle=\langle\mathds{B}|\mathds{A}\rangle$, and the notation $\langle  \cdot,\cdot \rangle $ is defined as the inner product, such that $\langle A,B \rangle:=\Tr{AB}$, and $d=|\set{X}||\set{A}|$. Note that each element in $\set{Z}$ corresponds to a unique vector in $\set{V}$.

Let us first define the convex cone generated by \textit{simple} PMDs as,
\begin{align}
&\mathcal{C}:= \{c\mathds{W}: c \in \mathbb{R}_{+},\mathds{W}\in \set{F}\},
\end{align}
as well as its dual,
\begin{align}
&\mathcal{C}^*:= \{\mathds{E} \in \set{V}: \langle\mathds{E}|\mathds{F}\rangle\geq 0, \forall \mathds{F}\in \set{C}\}.
\end{align}

We then define the generalized robustness of PMD $\mathds{M}$ with respect to $\set{F}$,
\begin{equation}\label{eq:defrb}
\mathfrak{R}(\mathds{M}):=\min \{r \in \mathbb{R}_{+}: \mathds{M}+r\mathds{N}\in\mathcal{C}, \quad \mathds{N}\in \mathcal{Z}\}
\end{equation}
where  $\mathds{M}+r\mathds{N}\in\mathcal{C}$ is equivalent to the fact that the family of $\{M(a|x)+rN(a|x)\}$ considered as a vector defined in~(\ref{vector}) is in set $\set{C}$. 

In order to see the connection between robustness and guessing game, let's us define $\mathds{N}':=r\mathds{N}$, i.e., $N'(a|x)=rN(a|x) \forall a,x$, and define $\mathds{N}' \succeq 0$ the same fashion as, $N'(a|x)\geq 0 \quad \forall a,x$, then we rewrite the definition~(\ref{eq:defrb}) as a conic form problem (which we call primal problem) with generalized inequality $\succeq$, i.e., given $\mathds{M}$, we want: 
\begin{equation}
\begin{array}{ll@{}ll}
\text{minimize}  & \displaystyle\lambda &\\
\text{subject to}& \displaystyle    \mathds{M}+\mathds{N}'\in\mathcal{C}&&\\
&\mathds{N}' \succeq 0, &&\\
&\sum_{a}N'(a|x)=\lambda\one, &&\forall x.
\end{array}
\end{equation}
Introducing hermitian operators $\gamma_x$ as Lagrange multiplies, we can write the Lagrangian with respect to $\mathds{M}$ as
\begin{align}
\mathcal{L}(\lambda,\mathds{N}', \mathds{A},\mathds{B},\{\gamma_{x}\})&=\lambda - \langle\mathds{M}+\mathds{N}'|\mathds{A}\rangle-\langle\mathds{N}'|\mathds{B}\rangle -\sum_{x}\Big \langle \lambda\one-\sum_{a}N'(a|x),\gamma_{x}\Big \rangle\\
&=-\langle\mathds{M}|\mathds{A}\rangle+\lambda(1-\sum_{x}\Tr{\gamma_{x}})+\sum_{a,x}\Big \langle N'(a|x),\gamma_{x}-\beta_{a,x}-\alpha_{a,x}\Big \rangle.
\end{align}
where the dual variables satisfy $ \mathds{A}\in \set{C}^*$, $ \mathds{B}\succeq 0$, and the elements of $\mathds{A}$  and $\mathds{B}$ are $\{\alpha_{a,x}\}$ and $\{\beta_{a,x}\}$ respectively.   Then we write the dual function as,
\begin{align}
g(\mathds{A},\mathds{B},\{\gamma_{x}\})&=\min_{\lambda,\mathds{N}'}\mathcal{L}(\lambda,\mathds{N}',\mathds{A},\mathds{B},\{\gamma_{x}\})\\
&=-\langle\mathds{M}|\mathds{A}\rangle+\min_{\lambda,\mathds{N}'}\bigg(\lambda(1-\sum_{x}\Tr{\gamma_{x}})+\sum_{a,x}\Big \langle N'(a|x),\gamma_{x}-\beta_{a,x}-\alpha_{a,x}\Big \rangle\bigg)
\end{align}
since $g$ is linear function and a linear function is bounded below only when it is identical zero. Thus, $g=-\infty$ (trivial bound), except only when the following two conditions hold,
\[
\begin{cases} 
\sum_{x}\Tr{\gamma_{x}}=1& \\
\gamma_{x}-\beta_{a,x}-\alpha_{a,x}=0 &\forall a,x,  
\end{cases}
\]
in which cases, $g(\mathds{A},\mathds{B},\{\gamma_{x}\})=-\langle\mathds{M}|\mathds{A}\rangle$.	
So we can write the dual problem to define the upper bound of dual fucntion as follows,
\begin{equation}
\begin{array}{ll@{}ll}
\text{maximize}  & \displaystyle -\langle\mathds{M}|\mathds{A}\rangle&\\
\text{subject to}& \displaystyle \mathds{A}\in  \set{C}^* &\\
&\mathds{B}\succeq 0,  &&\\
& 	\gamma_{x}-\beta_{a,x}-\alpha_{a,x}=0  &&\forall a,x,\\
&\sum_{x}\Tr{\gamma_{x}}=1,
&&\gamma_{x}=\gamma^{\dagger}_{x}.\\
\end{array}
\end{equation}
we can get rid of the dual variable $\mathds{B}$ by combining the second and third constriant as the condition $\gamma_{x}-\alpha_{a,x}\geq 0\quad\forall a,x$, becasue $\mathds{B}$ is only the constraint of dual variables, then the above problem reduces to,
\begin{equation}\label{sdp}
\begin{array}{ll@{}ll}
\text{maximize}  & \displaystyle -\langle\mathds{M}|\mathds{A}\rangle&\\
\text{subject to}& \displaystyle \mathds{A}\in  \set{C}^* &\\
& \gamma_{x}-\alpha_{a,x}\geq 0  &&\forall a,x,\\
&\sum_{x}\Tr{\gamma_{x}}=1,
&&\gamma_{x}=\gamma^{\dagger}_{x}.\\
\end{array}
\end{equation}
Define a new variable $\mathds{W}$, such that its element  $\omega_{a,x}:=\gamma_{x}-\alpha_{a,x}$, and we see that,
\[
-\langle\mathds{M}|\mathds{A}\rangle=\langle\mathds{M}|\mathds{W}\rangle-1=\sum_{a,x}\Big \langle M(a|x), \omega_{a,x} \Big \rangle-1,
\] 
then we can rewrite the dual problem as, 
\begin{equation}\label{sdp2}
\begin{array}{ll@{}ll}
\text{maximize}  & \displaystyle\langle\mathds{M}|\mathds{W}\rangle-1&\\
\text{subject to}& \displaystyle   
\mathds{A} \in\set{C}^*&\\
&\mathds{W} \succeq  0 &\\
&\sum_{x}\Tr{\gamma_{x}}=1, &&\gamma_{x}=\gamma^{\dagger}_{x}.\\
\end{array}
\end{equation}
where $\mathds{W} \succeq  0$ is equivalent to, $\gamma_{x}-\alpha_{a,x}\geq 0 \forall a,x$.
To see that the strong duality holds, that is to say,  the optimal value of the dual is equal to the optimal value of the primal problem, let's choose $\alpha_{a,x}=\frac{1}{2|\set{X}|\Tr{\one}}\one,  \forall a,x$, i.e., $\mathds{A} \succ 0$ (thus $\mathds{A}$ is in the interior of $\set{C}^*$), and $\gamma_{x}:=2\alpha_{a,x}$, we then see that $\gamma_{x}-\alpha_{a,x}= \alpha_{a,x} > 0\quad  \forall a,x$ and $\sum_{x}\Tr{\gamma_{x}}=1$. These choices can be noticed to strictly satisfy the conditions~(\ref{sdp2}). So Slater's theorem ensures that the strong duality holds.

\begin{theorem}\label{robustness}
	For any PMD, with its robustness related to guessing games, it satisfies
	\begin{equation}\label{maineq}
	1+\mathfrak{R}( \{M(a|x)\}_{a,x})=\max_{\{\rho_{a,x} \}} \frac{\pguess(M(a|x);\rho_{a,x})}{\pguess^{\text{simple}}(\rho_{a,x})}\;,
	\end{equation}
\end{theorem}

\begin{proof}
	We first show the right hand side is smaller than or equal to the left hand side for all possible ensembles, then we show a special choosing ensemble satisfies that the right hand side is greater than or equal to the left hand side, which can be seen as the optimal ensemble.
	
	According to the definition of the general robustness of $\{M(a|x)\}_{a,x}$, one can write $M(a|x)=(1+r)F(a|x)-rN(a|x)$ for some $\mathds{F}\in \set{F}$ with elements as $F(a|x)$, where $r=\mathfrak{R}( \{M(a|x)\}_{a,x})$. By using the same notations as shown in Theorem~\ref{th:main}, and according to Lemma~\ref{lemma:simple-all-equiv}, we obtain,
	\begin{align}
	&\pguess^{\text{simple}}(\rho_{a,x})=\max_{\mathcal{T}\in\mathscr{T}}\sum_{a,x}\Big \langle \mathcal{T}(F)(a|x), \rho_{a,x}, \Big \rangle
	\end{align}
	and,
	\begin{align}
	\pguess(M(a|x);\rho_{a,x})&=\max_{\mathcal{T}\in\mathscr{T}}\sum_{a,x}\Big \langle \mathcal{T}(M)(a|x), \rho_{a,x} \Big \rangle\\
	&=\sum_{a,x}\Big \langle \mathcal{T}^*(M)(a|x), \rho_{a,x} \Big \rangle\;\\
	&=(1+r)\sum_{a,x}\Big \langle \mathcal{T}^*(F)(a|x), \rho_{a,x} \Big \rangle-r\sum_{a,x}\Big \langle \mathcal{T}^*(N)(a|x), \rho_{a,x} \Big \rangle\\
	&\le (1+r)\max_{\mathcal{T}\in\mathscr{T}}\sum_{a,x}\Big \langle \mathcal{T}(F)(a|x), \rho_{a,x} \Big \rangle\\
	&=(1+r)\pguess^{\text{simple}}(\rho_{a,x}),
	\end{align}
	where the third equality holds because the optimized $\mathcal{T}^*$ is linear.
	
	Next we choose an ensemble $\rho_{a,x}=\omega_{a,x}$ (up to normalization constriant of the ensemble) satisfying the constraint in
	where we consider the set of optimal $\{\omega_{a,x}\}$ appear in the dual problem~(\ref{sdp2}), under this ensemble, we obtain, 
	\begin{align}
	\frac{\pguess(M(a|x);\rho_{a,x})}{\pguess^{\text{simple}}(\rho_{a,x})}&=\frac{\max_{\mathcal{T}\in\mathscr{T}}\sum_{a,x}\Big \langle \mathcal{T}(M)(a|x), \omega_{a,x} \Big \rangle}{\max_{\mathcal{T}\in\mathscr{T}}\sum_{a,x}\Big \langle \mathcal{T}(F)(a|x), \omega_{a,x} \Big \rangle}\\
	&\ge \frac{\sum_{a,x}\Big \langle M(a|x), \omega_{a,x} \Big \rangle}{\max_{\mathcal{T}\in\mathscr{T}}\sum_{a,x}\Big \langle \mathcal{T}(F)(a|x), \omega_{a,x} \Big \rangle}\\
	&\ge \frac{\sum_{a,x}\Big \langle M(a|x), \omega_{a,x} \Big \rangle}{\max_{\mathcal{T}\in\mathscr{T}}\sum_{a,x}\Big \langle \mathcal{T}(F)(a|x),\gamma_{x} \Big \rangle}\\ 
	&= \frac{\sum_{a,x}\Big \langle M(a|x), \omega_{a,x} \Big \rangle}{\sum_{x}\Big \langle \one,  \gamma_{x}\Big \rangle} \\
	&=1+\mathfrak{R}( \{M(a|x)\}_{a,x})\label{eq:linearity}
	\end{align}
	where the first inequality holds because of maximization over all possible $\mathcal{T}$, the second inequality holds because of the constraint  $\gamma_{x}-\omega_{a,x}=\alpha_{a,x}$  and $\mathds{A} \in \set{C}^*$ in~(\ref{sdp2}), which brings the fact that $\sum_{a,x}\Big \langle \mathcal{T}(F)(a|x),\gamma_{x} -\omega_{a,x}\Big \rangle  \ge 0$, and the last equality holds because we have that $\sum_{a}\mathcal{T}(F)(a|x)=\one $ and also $\sum_{x}\Tr{\gamma_{x}}=1$, which concludes the proof.
\end{proof}

\end{document}